\documentclass[journal]{IEEEtran}

\usepackage{ifpdf}
\ifpdf
\else
 \fi

\usepackage{cite}
\usepackage[colorlinks=true,linkcolor=black,citecolor=black,urlcolor=blue]{hyperref}

\ifCLASSINFOpdf
\usepackage[pdftex]{graphicx}
\graphicspath{{../pdf/}{../jpeg/}}
\DeclareGraphicsExtensions{.pdf,.jpeg,.png}
\else
\usepackage[dvips]{graphicx}
\graphicspath{{../eps/}}
\DeclareGraphicsExtensions{.eps}
\fi

\makeatletter

\newcommand{\Rmnum}[1]{\expandafter\@slowromancap\romannumeral #1@}
\makeatother

\usepackage{amssymb}
\usepackage{amsmath}
\newtheorem{theorem}{Theorem}

\newtheorem{remark}{Remark}

\newenvironment{proof}{{ \textbf{Proof:}}}{\hfill \par}

\usepackage{array}

\ifCLASSOPTIONcompsoc
 \usepackage[caption=false,font=normalsize,labelfont=sf,textfont=sf]{subfig}
\else
  \usepackage[caption=false,font=footnotesize]{subfig}
\fi

\usepackage{stfloats}

\ifCLASSOPTIONcaptionsoff
  \usepackage[nomarkers]{endfloat}
 \let\MYoriglatexcaption\caption
\renewcommand{\caption}[2][\relax]{\MYoriglatexcaption[#2]{#2}}
\fi
\let\MYorigsubfloat\subfloat
\renewcommand{\subfloat}[2][\relax]{\MYorigsubfloat[]{#2}}

\renewcommand{\figurename}{Fig. }
\usepackage{url}

\usepackage{tikz,xcolor,hyperref}
\definecolor{lime}{HTML}{A6CE39}
\DeclareRobustCommand{\orcidicon}{
    \begin{tikzpicture}
    \draw[lime, fill=lime] (0,0) 
    circle [radius=0.16] 
    node[white] {{\fontfamily{qag}\selectfont \tiny ID}};    \draw[white, fill=white] (-0.0625,0.095) 
    circle [radius=0.007];    \end{tikzpicture}
    \hspace{-2mm}}
\foreach \x in {A, ..., Z}{
    \expandafter\xdef\csname orcid\x\endcsname{\noexpand\href{https://orcid.org/\csname orcidauthor\x\endcsname}{\noexpand\orcidicon}}
    }

\begin{document}
\title{Robust Event Triggering Control for Lateral Dynamics of Intelligent Vehicles with Designable Inter-event Times}

\author{Xing Chu\orcidA{},~\IEEEmembership{Member,~IEEE,}
        Zhi Liu\orcidD{},
        Lei Mao,
        Xin Jin\orcidC{},~\IEEEmembership{Member,~IEEE},~\\
        Zhaoxia Peng\orcidB{},
        and~ Guoguang Wen\orcidE{},~\IEEEmembership{Member,~IEEE}
\thanks{This work is supported by the Youth Talent Project of Yunnan Province under Grant No.C619300A020, the National Natural Science Foundation of China under Grant No.62101481, and Key Laboratory in Software Engineering of Yunnan Province in China (No.2020SE408; No.2020SE317). (Corresponding author: Guoguang Wen.)}\thanks{Xing Chu, Zhi Liu, and Xin Jin are with the School of Software, Yunnan University, Kunming 650091, China (e-mail: zhi\_liu97@163.com; chx@ynu.edu.cn; xinxin\_jin@163.com)}\thanks{Lei Mao is with the traffic engineering design and research in Yunnan Communications Investment \& Construction Group Co. Ltd. (e-mail: 1191633510@qq.com)}\thanks{Zhaoxia Peng is with the School of Transportation Science and Engineering, Beihang University, Beijing 100191, China (e-mail: pengzhaoxia@buaa.edu.cn).}\thanks{Guoguang Wen is with the Department of Mathematics, Beijing Jiaotong University, Beijing 100044, China (e-mail:guoguang.wen@bjtu.edu.cn).}}

\markboth{}
{Shell \MakeLowercase{\textit{et al.}}: Bare Demo of IEEEtran.cls for IEEE Journals}
% You can use \ifCLASSOPTIONpeerreview for conditional compilation here if you desire.
% If you want to put a publisher's ID mark on the page you can do it like this:
% Remember, if you use this you must call \IEEEpubidadjcol in the second column for its text to clear the IEEEpubid mark.

\maketitle
\begin{abstract}
In this brief, an improved event-triggered update mechanism (ETM) for the linear quadratic regulator is proposed to solve the lateral motion control problem of intelligent vehicle under bounded disturbances. Based on a novel event function using a clock-like variable to determine the triggering time, we further introduce two new design parameters to improve control performance. Distinct from existing event-based control mechanisms, the inter-event times (IETs) derived from the above control framework are designable, meaning that the proposed ETM can be deployed on practical vehicle more easily and effectively. In addition, the improved IETs-designable ETM features a global robust event-separation property that is extremely required for practical lateral motion control of vehicle subject to diverse disturbances. Theoretical analysis proves the feasibility and stability of the proposed control strategy for trajectory tracking under bounded disturbances. Finally, simulation results verify the theoretical results and show the advantages of the proposed control strategy.
\end{abstract}


% For peer review papers, you can put extra information on the cover page as needed:
% \fi
\IEEEpeerreviewmaketitle

\section{INTRODUCTION}
\IEEEPARstart {I}{ntelligent} vehicle is a promising solution for improving the safety and efficiency of traffic in modern intelligent transportation systems. In addition, it makes valuable contributions to reducing air pollution and saving energy for improving the natural environment and social economy. Thereinto, lateral motion control, as a crucial problem in intelligent vehicle, still faces many challenges. This brief focuses on the lateral motion control of intelligent vehicle.

The core objective of lateral motion control is achieving autonomous driving along the desired path. This research field has been expanded and innovated continuously in recent years. The kinematics and dynamics of vehicular lateral motion control were discussed in \cite{rajamani2011vehicle}. Owing to the importance of lateral motion control in autonomous driving, diverse control approaches have been proposed and developed, including PID control \cite{marino2011nested}, $H_{\infty}$ control \cite{zhang2018vehicle,MemoryET}, linear quadratic regulator (LQR) \cite{riofrio2017lqr}, model predictive control (MPC) \cite{fal2007mpc,Rout2021MPC,zhang2021adaptive}, and sliding-mode control (SMC) \cite{tagne2013higher,zhang2020SMC}, and just name a few.

The above review shows rich results of lateral motion control in intelligent vehicle area. Nevertheless, after these sort of cutting-edge control algorithms are deployed in the real digital platforms of vehicles, several practical issues should be further considered. For example, the computing unit and the bandwidth of controller area network (CAN) bus are shared with other control loops and extremely limited. Thus, reducing the unnecessary control computing and data traffic so as to improve the efficiency of resource utilization and energy consumption, while guaranteeing the original control performance, is what we desperately need. Under this background, event-based control mechanism, 
as shown in \cite{tabuada2007et,tsi2010datadriven,heemels2012intro,girard2015dynamic,berneburg2019robust} and the references therein, is introduced to the recent varied control researches \cite{chu2018distributed,song2021ETfixtime,lh2021consensus,chu2019distributed}. Instead of feeding back all the sampling data regardless of the practical need (for example, stability) as in the current rich results, event-based control schemes determine the time instant, at which the data of control system is sampled and sent back to update the control input, based on certain triggering rules.

Unfortunately, there are still some drawbacks to dealing with real control problems under traditional event-triggered schemes. For example, choosing a larger threshold means tolerating a larger state error, which obviously conserves resources in such a control scheme, but also influences control performance, such as the accuracy of trajectory tracking. So far, hence, several researchers have studied adaptive event-based control schemes to solve this problem. In \cite{zhang2021adaptive}, adaptive event-based model predictive control was developed for the control of unmanned vehicles. An adaptive-event-trigger-based fuzzy nonlinear lateral dynamic control algorithm was proposed in \cite{li2021fuzzy}. On the other hand, there are still some issues that are less frequently mentioned, such as the practical constraints of inter-event times (IETs) imposed on the sensor and communication, which prevent the traditional event-triggered schemes from being applied flexibly in practical control systems. More importantly, global robust event-separation property of event-triggered control system, which determines the feasibility and reliability of the designed event triggering scheme, cannot be ignored.

A potential solution to cope with the all above-mentioned issues of event-based control mechanisms was provided in \cite{DesignableMIET}. Specifically, a novel design and analysis approach of an event function was provided, which guarantees original control performance through a clock-like triggering signal whose variable range and evolution rate can be flexibly adjusted to obtain the desired IETs. Meanwhile, the evet-triggered control studied in \cite{DesignableMIET} holds the global robust event-separation property. This type of IETs-designable event triggering mechanism (ETM) is therefore extremely suitable for the deployment on practical systems, which constructs the main motivation for this work.

In this brief, inspired by the idea from the IET-designable ETM, we aim to tackle the execution problem of lateral LQR of intelligent vehicle. The contributions of this brief are the following:
\begin{itemize}
\item An event-based execution mechanism for a LQR controller to solve the lateral motion control problem of intelligent vehicle is designed. Compared with the time-triggered strategy, this mechanism is proved to reduce unnecessary communication and control update, so as to reduce energy consumption and resource occupation, while stabilizing the system.

\item By introducing two new design parameters to improve the performance of the IETs-designable ETM, the IETs can be adjusted more freely than in \cite{DesignableMIET}. Thus, designability of the IETs is further expanded, which benefits its application in actual systems.

\item Many existing ETMs do not feature a robust event-separation property, so any tiny disturbance maybe leads to the Zeno-behavior. Considering the bounded disturbances in the real movement of intelligent vehicle, we verify the stability and feasibility of the proposed event-based controller with improved IETs-designable ETM.
\end{itemize}

\section{PROBLEM FORMULATION}

\subsection{VEHICLE DYNAMIC MODEL DESCRIPTION}
In the dynamic bicycle model illustrated in \figurename \ref{fig:Model}, it is assumed that the vehicle is assumed to be symmetrical, and tire sideslip angles on the same axle are equal. We neglect the roll and pitch dynamics, and assume the sideslip $\beta$, yaw $\psi$, and control input steering angle $\delta$ are small. A linear parameter variation (LPV) model is derived from the linear tire force model. \cite{rajamani2011vehicle} provides dynamic equations for the bicycle model in terms of sideslip angle and yaw rate.
\begin{equation}
\begin{cases}\begin{array}{l}\label{beta&psi}
\dot{\beta}=-\frac{\mu\left(C_{f}+C_{r}\right)}{m V_{x}} \beta-\left(1+\frac{\mu\left(l_{f} C_{f}-l_{r} C_{r}\right)}{m V_{x}^{2}}\right) \dot{\psi}+\frac{\mu C_{f}}{m V_{x}} \delta \vspace{1ex} \\
\ddot{\psi}=-\frac{\mu\left(l_{f} C_{f}-l_{r} C_{r}\right)}{I_{z}} \beta-\frac{\mu\left(l_{f}^{2} C_{f}+l_{r}^{2} C_{r}\right)}{I_{z} V_{x}} \dot{\psi}+\frac{\mu l_{f} C_{f}}{I_{z}} \delta
\end{array}
\end{cases}.
\end{equation}
The corresponding vehicle parameters and nomenclature are listed in Table \ref{Parameters}. 

\begin{figure}
    \centering
    \includegraphics[width=6cm]{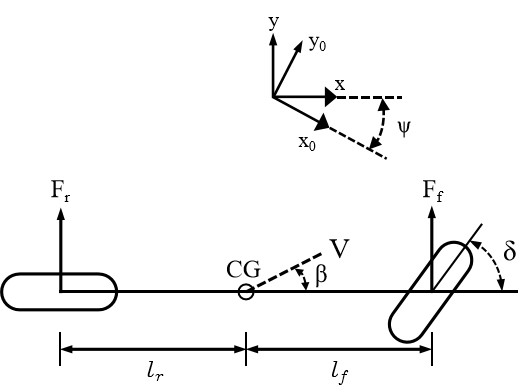}
    \caption{Dynamic bicycle model.}
    \label{fig:Model}
\end{figure}

\begin{table}[]
\centering
\caption{\centering\textbf{PARTIAL PARAMETERS AND NOMENCLATURE}}
\label{Parameters}
\setlength{\tabcolsep}{8mm}{
\begin{tabular}{cc}
\hline \hline
Symbol     & Description                             \\ \hline
$\beta$      & Vehicle sideslip angle                        \\ %\hline
$\dot{\psi}$ & Yaw rate                              \\ %\hline
$\delta$       & Steering wheel angle                \\ %\hline
$\mu$        & Road friction coefficient             \\ %\hline
$V_x$        & Vehicle longitudinal velocity                 \\ %\hline
$m$          & Vehicle mass                                  \\ %\hline
CG          & Center of gravity of vehicle           \\ %\hline
$I_z$        & Yaw moment inertia of vehicle                 \\ %\hline
$l_f$        & Front axle-CG distance                \\ %\hline
$l_r$        & Rear axle-CG distance                 \\ %\hline
$C_f$        & Front tire cornering stiffness  \\ %\hline
$C_r$        & Rear tire cornering stiffness   \\ %\hline
\hline
\hline
\end{tabular}}
\end{table}

\subsection{CONTROL PROBLEM DEFINITION}
The objective of lateral control is to minimize a vehicle's lateral displacement relative to a specific reference path. The dynamics of a vehicle's lateral error displacement $e$ relative to a reference path is $\ddot{e}=V_x(\dot{\beta}+\dot{\psi})-V^2_x \rho$,
where the road curvature $\rho$ is a constant. And based on (\ref{beta&psi}), system state variables can be denoted $x=[\beta, \dot{\psi}, \dot{e}, e]^{\mathrm{T}}$, and the control input is the steering angle $\delta$. Then, owing to eliminating the lateral error, for a given road curvature $\rho$ and longitudinal velocity $V_x$, $\dot{e}=e=0$ is the control object. In \cite{rajamani2011vehicle}, the desired equilibrium point can be written as $x^\star=[\beta, \dot{\psi}, \dot{e}, e]^{\mathrm{T}}=[\beta^\star, \dot{\psi}^\star,0,0]^{\mathrm{T}}\nonumber$, and the control input $\delta$ at the desired equilibrium point is now $\beta^{\star}$, where $\beta^{\star}=\left(L_{r}-\dfrac{L_{f} m V_{x}^{2}}{\mu C_{r}\left(L_{f}+L_{r}\right)}\right) \rho$, 
$\dot{\psi}^{\star}=V_{x} \rho$,
and $\delta^{\star}=\left(L_{f}+L_{r}\right) \rho+\dfrac{m V_{x}^{2}\left(L_{r} C_{r}-L_{f} C_{f}\right)}{\mu C_{f} C_{r}\left(L_{f}+L_{r}\right)} \rho$.

By defining the new error variables $\tilde{\beta}=\beta-\beta^{\star}$, $\dot{\tilde{\psi}}=\dot{\psi}-\dot{\psi}^{\star}$, and $\tilde{\delta}=\delta-\delta^{\star}$, the new system state variables become $\tilde{x}=(\tilde{\beta}, \dot{\tilde{\psi}}, \dot{e}, e)^{\mathrm{T}}$. We write the error dynamics \cite{Tagne2016DesignAndComparison}, the origin of which is the equilibrium, as $\tilde{x}^{\star}=(\tilde{\beta}, \dot{\tilde{\psi}}, \dot{e}, e)^{\mathrm{T}}=(0,0,0,0)^{\mathrm{T}}$ as:
\begin{equation}\label{errorDynamics}
    \dot{\tilde{x}}(t)= A \tilde{x}(t)+B \tilde{\delta}(t),
\end{equation}
where
\begin{equation}\label{AB}
\begin{split}\begin{array}{l}
A=\left[\begin{array}{cccc}
-\frac{\mu\left(C_{f}+C_{r}\right)}{m V_{x}} & -1-\frac{\mu\left(l_{f} C_{f}-l_{r} C_{r}\right)}{m V_{x}^{2}} & 0 & 0  \vspace{1ex} \\
-\frac{\mu\left(l_{f} C_{f}-l_{r} C_{r}\right)}{I_{z}} & -\frac{\mu\left(l_{f}^{2} C_{f}+l_{r}^{2} C_{r}\right)}{I_{z} V_{x}} & 0 & 0 \vspace{1ex} \\
-\frac{\mu\left(C_{f}+C_{r}\right)}{m} & -\frac{\mu\left(l_{f} C_{f}-l_{r} C_{r}\right)}{m V_{x}} & 0 & 0 \vspace{1ex} \\
0 & 0 & 1 & 0
\end{array}\right], \vspace{1ex} \\ 
B=\left[\begin{array}{cccc}
\frac{\mu C_{f}}{m V_{x}} ;& \frac{\mu l_{f} C_{f}}{I_{z}};& \frac{\mu C_{f}}{m};&  0
\end{array}\right]^{\mathrm{T}} .
\end{array}\end{split}
\end{equation}

\begin{remark}
The steady-state solution is the desired reference state $x^\star$ in this model. As a result, $\tilde{\delta}$ can ignore the model's natural transient dynamics when using this as a reference input. Additional information about this section can be found in \cite{rajamani2011vehicle,Tagne2016DesignAndComparison,snider2009automatic}. 
\end{remark}

\subsection{LQR CONTROLLER}
To solve the control problem in (\ref{errorDynamics}), the LQR controller is designed. The linear quadratic cost functional index is constructed first:
\begin{equation}\label{indexJ}
    J=\int_{t_{0}}^{\infty}\left[\tilde{x}^{\mathrm{T}}(t) Q \tilde{x}(t)+\tilde{\delta}^{\mathrm{T}}(t) R \tilde{\delta}(t)\right] \mathrm{d} t.
\end{equation}
Then, the control input steering angle $\tilde{\delta}$ is calculated according to
\begin{equation}\label{feedback}
\tilde{\delta}(t)=-K \cdot \tilde{x}(t)=-R^{-1}B^{\mathrm{T}} P \cdot \tilde{x},
\end{equation}
where $P$ is derived by solving the Riccati equation:
\begin{equation}\label{Riccati}
    A^{\mathrm{T}} P+P A-{P B R}^{-1} B^{\mathrm{T}} P+Q=0.
\end{equation}

\section{MAIN RESULTS}

\subsection{IET-DESIGNABLE EVENT TRIGGERING MECHANISM}
In this subsection, we aim to apply the improved IET-designable ETM to determine the triggering instant when the feedback control law is updated, which is a malleable solution to the computational load and hardware limitations. This mechanism is based on a countdown variable with an upper bound. Thus, we create an event function $\mathbf{Z}(t)$ whose upper bound $\mathbf{\bar{Z}}$ can be freely set and apply the following event triggering rule:
\begin{equation}\label{triggerRule}
    \begin{aligned}
t_{0} &=0 ,\\
t_{k+1} &=i n f\left\{t>t_{k} \mid \mathbf{Z}(t)= 0\right\},
\end{aligned}
\end{equation}
where $\mathbf{Z}(t)$ is reset as the design parameter $\mathbf{\bar{Z}}>0$ at every triggering instant. The triggering rule (\ref{triggerRule}) will determine the time series $\left\{t_{k}\right\}, k \in {\mathbb{N}^{*}}$. To ensure the feasibility and stability of the IET-designable ETM, the event function dynamics is defined as $\dot{\mathbf{Z}}(t)=\omega(\varpi, \varepsilon)$, where $\varepsilon>0$ is also a design parameter. Detailed descriptions of the event function are given in the following subsection.

\subsection{STABILITY AND FEASIBILITY ANALYSIS}
In the movement process of intelligent vehicle, an inevitable disturbance is necessary, especially in the case of road curvature $\rho>0$. After adding the disturbances, the system dynamics (\ref{errorDynamics}) become
\begin{equation}\label{errorPERDynamics}
    \dot{\tilde{x}}(t)=A \tilde{x}(t)+B \tilde{\delta}(t)+G\xi(t),
\end{equation}
where $A$ and $B$ are shown in (\ref{AB}), and for all initial conditions $\tilde{x}(0)$ any bounded disturbance $\left | \xi \right | \le \bar{\xi}$, $G$ denotes a corresponding constant matrix. Moreover, we assume that the disturbance $\xi(t)$ is convergent to $0$. 

To stabilize the system (\ref{errorPERDynamics}), recalling (\ref{feedback}), the feedback control law is $\tilde{\delta}(t)=-K \tilde{x}(t)$, and the closed-loop control system considering the perturbation $\xi(t)$ can be written as $\dot{\tilde{x}}(t)= (A-BK) \tilde{x}(t)+G\xi(t)$. If a Lyapunov function $V = \tilde{x}^{\mathrm{T}} M \tilde{x}$ exists, where $M$ is a symmetric positive definite matrix,
\begin{equation}\label{LyapPQ}
    (A-BK)^{\mathrm{T}} M+M(A-BK)=-N
\end{equation}
is satisfied, where $N$ is an arbitrary symmetric positive-definite matrix. Under the IET-designable ETM (\ref{triggerRule}), the control input becomes 
\begin{equation}
    {\tilde{\delta }}(t)=-K \tilde{x}(t_{k}), \, t \in[t_{k},t_{k}+1).
\end{equation}
Then, define the event triggering error and its derivative as
\begin{equation}\begin{aligned}\label{eventError}
    \eta(t) &= \tilde{x}(t_{k})-\tilde{x}(t), \, t \in\left[t_{k}, t_{k+1}\right)  \\
    \dot{\eta}(t) &= -\dot{\tilde{x}}(t).
\end{aligned}
\end{equation}
We can the express the system (\ref{errorPERDynamics}) as
\begin{equation}\label{Lastsystem}
        \dot{\tilde{x}}(t)= A\tilde{x}(t)-BK \tilde{x}(t) -BK\eta(t) +G\xi(t).
\end{equation}
\begin{theorem}\label{theorem-1}
For all initial conditions $x(0)$, applying the IET-designable ETM (\ref{triggerRule}) in the system (\ref{Lastsystem}), $\tilde{x}(t)$ asymptotically converges to $\tilde{x}^{\star}$. Besides, designable IETs with lower-bounded $\tau$ are available:
\begin{equation}\label{MIET}
    \begin{aligned}
\tau &=\sqrt{\frac{1}{\sigma \varepsilon}}\left\{\operatorname{atan}\left[\sqrt{\frac{\sigma}{\varepsilon}}(1+\mathbf{\bar{Z}})\right]-\operatorname{atan}\left[\sqrt{\frac{\sigma}{\varepsilon}}\right]\right\}>0  \vspace{1ex}\\
\sigma&=\frac{\theta_{r}^{2}|M B K|^{2}}{\theta_{l}\lambda_{\min }(M) \lambda_{\min }(N)}.
\end{aligned}
\end{equation}
\end{theorem}

\begin{proof}
We construct the Lyapunov function candidate $W=\frac{1}{2} \tilde{x}^{\mathrm{T}} M \tilde{x}+ \frac{1}{2} \mathbf{Z} \eta^{\mathrm{T}}M\eta$. Based on (\ref{LyapPQ}), (\ref{eventError}), and (\ref{Lastsystem}), its derivative can be written as
\begin{eqnarray}\label{dotWeq}
\dot{W}&=& \tilde{x}^\mathrm{T}M\dot{\tilde{x}}+\frac{1}{2}\omega\eta^{\mathrm{T}}M\eta+\mathbf{Z}\eta^\mathrm{T}M\dot{\eta}                                                                 \nonumber  \\
       &=&-\frac{1}{2}\tilde{x}^ \mathrm{T} N \tilde{x}  - \tilde{x}^\mathrm{T}MBK\eta + \tilde{x}MG\xi    +\frac{1}{2} \omega  \eta ^\mathrm{T}M\eta                                  \nonumber  \\
       &\;& - \mathbf{Z}\eta^\mathrm{T} M(A-BK)\tilde {x} + \mathbf{Z}\eta^\mathrm{T}MBK\eta -\mathbf{Z}\eta^ \mathrm{T}M\xi.                                             \nonumber 
\end{eqnarray}
By some properties of matrix and norm, we can obtain
\begin{eqnarray}\label{dotWineq}
\dot{W}
   &\le& -\frac{1}{2} \lambda_{\min}(N)|\tilde{x}|^2+|\tilde{x}||MBK||\eta|+\frac{1}{2}\omega\lambda_{\min}(M)|\eta|^2                                                                      \nonumber \\
   &\;& +\mathbf{Z}|\eta||MA||\tilde{x}|+\mathbf{Z}|\eta||MBK||\tilde {x}| +\mathbf{Z}|MBK||\eta|^2                                                                                          \nonumber  \\
   &\;& +\mathbf{Z}|\eta ||M||\xi| +|\tilde{x}||MG||\xi|                             \nonumber,  
\end{eqnarray}
where $\left | \cdot  \right |$ denotes the Euclidean norm for vectors and the induced 2-norm for matrices. We define the variable $\varpi$ as
\begin{equation}\label{varphi}
    \varpi=\frac{\theta_{l}\lambda_{\min}(N)}{\lambda_{\min }(M)} \frac{|\tilde{x}|^{2}}{|\eta |^{2}}-2(1+\mathbf{Z})\frac{\theta_{r} |M B K|}{\lambda_{\min }(M)} \frac{|\tilde{x}|}{|\eta |},
\end{equation}
where $\theta_{l}>1$ and $\theta_{r}<1$ are two new design parameters. Therefore, variable $\omega$ further satisfies the inequality $\omega < \varpi$ to guarantee the asymptotic stability of the closed-loop system (\ref{Lastsystem}). Using Young inequality for $2(1+\mathbf{Z}) \frac{\theta_{r}|M B K|}{\lambda_{\min }(M)} \frac{|\tilde{x}|}{|\eta|}$, we obtain the following inequality:
\begin{equation}
    -2(1+\mathbf{Z}) \frac{\theta_{r}|M B K|}{\lambda_{\min }(M)} \frac{|\tilde{x}|}{|\eta|} \ge -\sigma(1+\mathbf{Z})^2-\frac{{\theta_{r}^{2}|MBK|}^2}{\sigma{\lambda_{\min}}^2(M)}, \frac{{|\tilde{x}|}^2}{|{\eta|}^2}, \nonumber
\end{equation}
which can lead to (\ref{varphi}) arriving at
\begin{equation}
    \varpi \ge -\sigma(1+\mathbf{Z})^{2}+\left(\frac{\theta_{l}\lambda_{\min}(N)}{\lambda_{\min }(M)}-\frac{\theta_{r}^{2}|M B K|^{2}}{\sigma \lambda_{\min}^{2}(M)}\right) \frac{|\tilde {x}|^{2}}{|\eta |^{2}}. \nonumber
\end{equation}
By choosing $\sigma=\frac{\theta_{r}^{2}|M B K|^{2}}{\theta_{l}\lambda_{\min }(M) \lambda_{\min }(N)}$, we have
\begin{equation}
    \varpi-\varepsilon \ge -\sigma(1+\mathbf{Z} )^{2}-\varepsilon, \nonumber
\end{equation}
where $\varepsilon > 0$. Next, we can further design $\omega$ as follows:
\begin{equation}\label{omega}
    \omega=\left\{\begin{array}{ll}
\min (0, \varpi)-\varepsilon, & \eta \neq 0, \\
-\varepsilon, & \eta=0.
\end{array}\right.
\end{equation}
Analyzing the two instances of value of $\omega$, when $\eta \neq 0 $, if $\varpi < 0 $, it can always be obtained that $\omega=\varpi - \varepsilon \ge -\sigma(1+\mathbf{Z})^2-\varepsilon$, and if $\varpi \ge 0$,  $\omega= - \varepsilon \ge -\sigma(1+\mathbf{Z})^2-\varepsilon$. In addition, when $\eta=0$, the inequality also satisfies $\omega= - \varepsilon \ge -\sigma(1+\mathbf{Z})^2-\varepsilon$. Recalling the dynamics of event function $\dot{\mathbf{Z}}(t)=\omega(\varpi, \varepsilon)$, we have $\dot{\mathbf{Z}} \ge -\sigma(1+\mathbf{Z})^2-\varepsilon$ in all instances. By defining $\phi \le \mathbf{Z}$, where $\phi$ is the solution of $\dot{\phi} \ge -\sigma(1+\phi)^2-\varepsilon$ satisfying $\phi(0)=\mathbf{\bar{Z}}$, the IET is lower-bounded by the time $\tau$ takes for $\phi$ to evolve from $\mathbf{\bar{Z}}$ to 0. $\hfill\blacksquare$
\end{proof}

Compared with the results in \cite{DesignableMIET}, the two new design parameters $\theta_l$ and $\theta_r$ are introduced. As a result, the evolution rate of the event function $\mathbf{Z}(t)$ can be further adjusted, meaning that the IETs can be adjusted more freely.

It is worth noting that the derivation of minimum IET $\tau$ in (\ref{MIET}) is independent of the state $\tilde{x}$ or the event triggering error $\eta$, thus we can conclude that the IET-designable ETM features the global robust event-separation property. Meanwhile, solving a strictly positive $\tau$ automatically excludes Zeno behavior.

\begin{remark}
Theorem \ref{theorem-1} and its proof give an improved IET-designable ETM for the lateral dynamics of intelligent vehicle. We can find that a freely designable minimum IET relies on design parameters $\bar{\mathbf{Z}}$ and $\varepsilon$. Additionally, a desired minimum IET can be acquired by adjusting the parameter $\sigma$, which suggests that we can also adjust the matrices $M$ and $N$, control gain $K$, and parameters $\theta_l$ and $\theta_r$. 
\end{remark}
\begin{remark}
If considering that the disturbances $\xi$ is unnecessary, we can obtain the same results as in Theorem \ref{theorem-1}. Therefore, we conclude that the derivation process and the strictly positive minimum IET $\tau$ is guaranteed whether or not disturbances $\xi(t)$ exists. Furthermore, convergent $\xi(t)$ indicates convergent $\tilde{x}(t)$ as a result of the exponential stability of the linear system.
\end{remark}

\section{SIMULATION RESULTS}
In simulation, the entire simulation time is $15$ s and the sampling period is $0.01$ s. We adopt three different strategies to achieve the control object: an improved IET-designable ETM ($\theta_l=8$ and $\theta_r=0.1$), IET-designable ETM in \cite{DesignableMIET} (without $\theta_l$ and $\theta_r$), and time triggering (0.01 s). In addition, in an IET-designable ETM, $\bar{\mathbf{Z}}=1$ and $\varepsilon=1$ are chosen. By \cite{Tagne2016DesignAndComparison}, some values of bicycle model parameters are given as follows: $m=1421\mathrm{~kg}$, $\mu=0.6$, $V_{x}=18\mathrm{m} / \mathrm{s}$, $I_{z}=2570\mathrm{~kg} \cdot \mathrm{m}^{2}$, $C_{f}=170,550 N/rad$, $C_{r}=137,844 N/rad$, $l_{f}=1.191 m$, and $l_{r}=1.513 m$.
We set the initial state $\tilde{x}=[0;0;0;0]$ and choose the weight matrix $Q=\operatorname{diag}\{30,10,1,1\}$ and $R=[1000]$ of the LQR controller. Based on $A$ and $B$ in (\ref{AB}), assuming the matrix $N$ is an identity matrix, we obtain the matrix $M$ by (\ref{LyapPQ}). We next assume the upper-boundedness of the disturbance $\bar{\xi}=[3\times 10^{-4};1\times 10^{-3};0;0]$ and $G=\operatorname{diag}\{1,1,1,1\}$.

\begin{figure}
    \centering
    \includegraphics{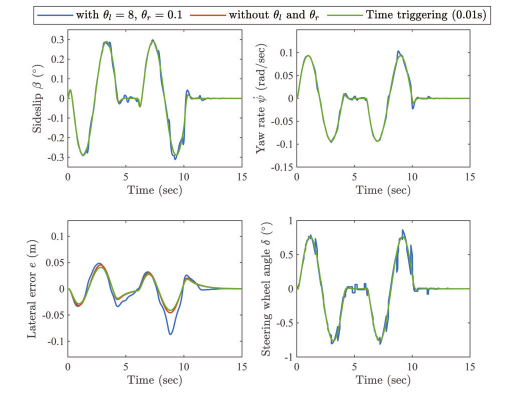}
    \caption{Sideslip $\beta$, yaw rate $\dot{\psi}$, lateral error displacement $e$, and control input steering angle $\delta$.}
    \label{fig:states}
\end{figure}

\begin{figure}
    \centering
    \includegraphics{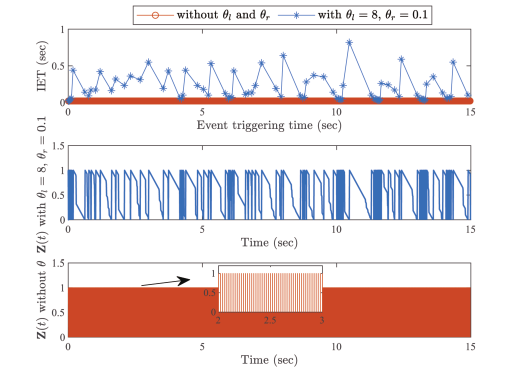}
    \caption{Countdown variable $\mathbf{Z}(t)$ and IET.}
    \label{fig:CountdownAndIET}
\end{figure}

\begin{figure}
    \centering
    \includegraphics{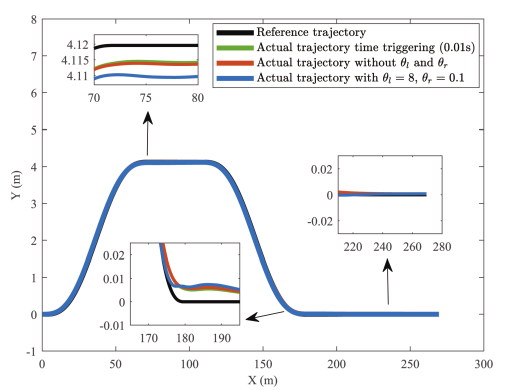}
    \caption{Desired and actual trajectories under different control schemes.}
    \label{fig:trajectory}
\end{figure}

The simulation results of sideslip $\beta$, yaw rate $\dot{\psi}$, lateral error $e$, and control input $\delta$ are computed according to the system state variables $\tilde{x}$ and desired equilibrium point $x^{\star}$, and are presented in \figurename \ref{fig:states}. \figurename \ref{fig:trajectory} displays the actual trajectories of three control strategies and a reference trajectory. \figurename \ref{fig:CountdownAndIET} depicts the countdown variable $\mathbf{Z}$ and IET at every triggering instant and implies the influence of the two new parameters. If we adopt the IET-designable ETM in \cite{DesignableMIET}, the value of $\varpi$ will be too large or too small, with the maximum being approximately $8\times10^8$. Furthermore, this means that we will acquire the same results for values of $\omega$ and $\mathbf{Z}$ as shown in Fig. \ref{fig:CountdownAndIET}. Obviously, this is a heavy computational burden, which contradicts the purpose of this work.

\begin{table}[]
\centering
\caption{\centering\textbf{COMPARISON OF TRIGGERING NUMBERS}}
\label{table:CompareTirggers}

\begin{tabular}{cccc}
\hline \hline
Strategy     & Triggering numbers                      \\  \hline 
Time triggering      & 1500                           \\  %\hline
IET-designable ETM in \cite{DesignableMIET} & 749    \\  %\hline
IET-designable ETM with $\theta_l=8$ and $\theta_r=0.1$      &83   \\ \hline
\hline
\end{tabular}
\end{table}

We attempted to change the values of design parameters $\bar{\mathbf{Z}}$ and $\varepsilon$, and chose different matrices $M$, $N$, $Q$, and $R$, but these attempts had little effect on the results, implying that the designability of IET-designable ETM in \cite{DesignableMIET} still has some limitations when applied in actual systems. To obtain a desirable control performance, we introduced two new design parameters, $\theta_l\ge1$ and $\theta_r\le1$. This is a natural operation that can ensure that $\omega$ further satisfies $\omega<\varpi$, so that the asymptotic stability of system is guaranteed. We adopted this improved IET-designable ETM, and the two new design parameters were chosen as $\theta_l=8$ and $\theta_r=0.1$. In Fig. \ref{fig:CountdownAndIET} and Table \ref{table:CompareTirggers}, it can be easily seen that the triggering numbers decrease significantly after introducing $\theta_l$ and $\theta_r$, leading to nearly up to a $94\%$ savings of triggering numbers compared to the time triggering control strategy and $88\%$ compared to the IET-designable ETM strategy in \cite{DesignableMIET}. 

It should be noted that when we select different parameters of vehicle models, some fluctuations in the results of saving triggering numbers occur, but this approach can save the triggering numbers by more than $60\%$ by adjusting the values of two new parameters $\theta_l$ and $\theta_r$, in general.

\section{CONCLUSIONS}
This brief focuses on the lateral motion control problem of intelligent vehicle via a malleable event triggering control mechanism for the LQR controller. Simulation results verify the performance gain from saving computational load and communication resources, while achieving the control objective under bounded disturbance. Furthermore, the two new design parameters greatly improve the designability of the IET-designable ETM, which can be developed further to broaden its applicability. In our future work, the dynamical characteristics of the IETs in the proposed IET-designable ETM will be further investigated.

\bibliographystyle{IEEEtran}
\bibliography{ref}
\vfill
\end{document}